\documentclass[envcountreset,envcountsect,runningheads]{llncs}
\usepackage{amsmath,amsfonts,amstext,amssymb,latexsym,epic,eepic,url}
\usepackage[pdftex]{graphicx}
\usepackage[mathscr]{eucal}
\usepackage{multirow}

\usepackage{tabularx} 
\newcolumntype{C}{>{\centering\arraybackslash}X} 
\usepackage{caption}
\usepackage{subcaption}
\usepackage{array}
\usepackage{framed}
\usepackage{xcolor}

\usepackage{cite}

\usepackage{multirow,bigdelim}

\usepackage{algorithmicx}
\usepackage[ruled]{algorithm}
\usepackage{algpseudocode}

\usepackage[framemethod=tikz]{mdframed}
\usepackage{color,soul,ulem}

\usetikzlibrary{arrows}

\usepackage{multirow}

\usepackage{mathrsfs}

\usepackage[ margin=1.5in]{geometry}

\usepackage{float}

\newcommand{\fd}{\rightarrow}

\newcommand{\spmvd}{\twoheadrightarrow_{sp}}
\newcommand{\mvd}{\twoheadrightarrow}
\newcommand{\spfd}{\rightarrow_{sp}}
\newcommand{\spk}[1]{sp\langle #1 \rangle}
\newcommand{\nul}{\texttt{NULL}}

\usepackage{scalerel,amssymb}

\usepackage{tikz}

\usetikzlibrary{intersections}

\usetikzlibrary{decorations.markings}

\usetikzlibrary{chains,fit,shapes,decorations.pathreplacing}

\tikzstyle{vertex}=[circle, draw, inner sep=0pt, minimum size=6pt]

\title{Approximate Integrity Constraints in Incomplete Databases With Limited Domains\thanks{Research of the second author was partially supported by the
    National Research, Development and Innovation Office (NKFIH)
    grants K--116769 and SNN-135643. This work was also   supported by the BME-
Artificial Intelligence FIKP grant of EMMI (BME FIKP-MI/SC) and by the Ministry of Innovation and
Technology and the National Research, Development and Innovation
Office within the Artificial Intelligence National Laboratory of Hungary.
}}
\toctitle{Approximate Integrity Constraints in Incomplete Databases With Limited Domains}
\titlerunning{Approximate Keys and FDs}
\author{Munqath Al-atar\inst{1}, Attila Sali\inst{2,}\inst{3}}
\institute{ITRDC,
University of Kufa\\ {\tt munqith.alattar@uokufa.edu.iq}\and
Department of Computer Science and Information Theory,\\ Budapest University of Technology and Economics \and
Alfr\'ed R\'enyi Institute of Mathematics\\
{\tt sali.attila@renyi.hu}}

\begin{document}
\maketitle
\setlength{\tabcolsep}{8pt}
\renewcommand{\arraystretch}{1.4}

\begin{abstract}
In case of incomplete database tables, a possible world is obtained by replacing any missing value by a value from the corresponding attribute's domain that can be infinite. A possible key or possible functional dependency constraint is satisfied by an incomplete table if we can obtain a possible world that satisfies the given key or functional dependency. On the other hand, a certain key or certain functional dependency holds if all possible worlds satisfy the constraint, A strongly possible constraint is an intermediate concept between possible and certain constraints, based on the strongly possible world approach (a strongly possible world is obtained by replacing \nul's by a value from the ones appearing in the corresponding attribute of the table).
A strongly possible key or functional dependency holds in an incomplete table if there exists a strongly possible world that satisfies the given constraint. 
In the present paper, we introduce strongly possible versions of multivalued dependencies and cross joins, and we analyse the complexity of checking the validity of a given strongly possible cross joins.
We also study approximation measures of strongly possible keys (spKeys), functional dependencies (spFDs), multivalued dependencies (spMVDs) and cross joins (spCJs). $g_3$ and $g_5$ measures are used to measure how close a table $Y$ satisfies a constraint if it is violated in $T$. Where the two measures $g_3$  and $g_5$ represent the ratio of the minimum number of tuples that are required to be removed from or added to, respectively, the table so that the constraint holds. Removing tuples may remove the cases that caused the constraint violation and adding tuples can extend the values shown on an attribute.  
For spKeys and spFDs, We show that the $g_3$ value is always an upper bound of the $g_5$ value for a given constraint in a table. However, there are tables of arbitrarily large number of tuples and a constant number of attributes that satisfy $g_3-g_5=\frac{p}{q}$ for any rational number $0\le\frac{p}{q}<1$. 
On the other hand, we show that the two measures values are independent of each other in the case of spMVDs and spCJs.

We also treat complexity questions of determination of the approximation values.

\end{abstract}


\section{Introduction}
Missing values is a common issue in many databases because of many reasons such as data repair, error while data transmission, data storage defect, etc \cite{farhangfar2007novel}.
Imputing a value instead of any missed information is the most common way o handle this issue \cite{lipski1981databases}.

In \cite{alattar2019keys}, a new imputation approach was introduced that considers only the active domain for each attribute (the data that are already there in that attribute of the table). The resulting complete table using this imputation approach is called a \emph{strongly possible world}. Using only the shown data for imputation ensures that we always use values related to the attributes domains.

Using this concept, new versions of integrity constraints called strongly possible constraints, such as keys (spKeys), functional dependencies (spFDs), multivalued dependencies (spMVDs) and cross joins (spCJs) were defined in \cite{alattar2020strongly,alattar2020functional,alattaratila_mvd} or in case of spCJs are introduced in the present paper,  that are satisfied after replacing any missing value (\nul) with a value that is already shown in the corresponding attribute. 
In section \ref{defs}, we provide the formal definitions.

The present paper is the extended version of the conference paper \cite{al2022approximate} that  continued the work started in  \cite{alattar2020strongly}, where an approximation approach to measure how close a table satisfies a key constraint was introduced. The closeness is measured by calculating the ratio of the minimum number of tuples to be removed so that the constraint holds in the resulting table. This is necessary as the active domains may not contain enough values to replace the \nul\ so that the resulting table satisfies the key.

Here we also consider strongly possible multivalued dependencies and cross joins.
We introduce approximation measures of spKeys, 
spFDs and spCJs by adding tuples. Adding a tuple with new unique values will add more values to the attributes' active domains, thus some unsatisfied equality generating constraints may get satisfied. On the other hand, adding tuples consisting of \nul\ values give enough tuples for tuple generating constraints. Note that if possible keys or functional dependencies \cite{kohler2016possible} are considered, then adding tuples does not change the satisfaction of the given constraint, as in that case any domain values could be used to replace \nul's.

For example, Table \ref{add_vs_rmv} does not satisfy the spKey $\spk{Car\_Model, Door No}$, while they are designed to distinguish Cars Types. We need to either remove two tuples or add a new tuple with new door number value to satisfy $\spk{Car\_Model, Door No}$.
Furthermore,car model and door number should, generally, determine the engine type, so, the added tuple with a new value in the $Door No$ attribute will satisfy $(Car\_Model, Door No) \spfd Engine\_Type$ instead of removing other two tuples.
\begin{table}[ht]
    \centering
\begin{tabular}{ c c c }
		\hline
		\textbf{Car\_Model} & \textbf{Door No} & \textbf{Engine\_Type}  \\  \hline
		 BMW & 4 & $\bot$ \\ 
		 BMW & $\bot$ & electric \\
		 Ford & $\bot$ & V8\\
		 Ford & $\bot$ & V6\\
		 \hline
	\end{tabular}
    \caption{Cars Information Incomplete Table}
    \label{add_vs_rmv}
\end{table}

If the key or the FD is not satisfied in the total part of the table, then adding a new values is useless. So, in this paper we suppose that $\spk{K}$ or $X\spfd Y$ is satisfied by the $K$-total part of the table (for exact definitions see Section~\ref{defs}).
In this paper, we assume that there is at least one non-null value for each attribute to have a non-empty active domain. 
The main objectives of this paper are:

\begin{itemize}
    \item Extend the $g_3$ measure defined for spKeys in \cite{alattar2020strongly} to other strongly possible integrity constraints. 
    \item Propose the $g_5$ approximation measure for strongly possible constraints, that measures the minimum number of tuples that are required to be added to satisfy the constraints. 
    \item Compare $g_5$ with $g_3$ and show that for spKeys and spFDs it is more effective to adding new tuples than removing violating ones.
    \item We show that except for the inequality above, $g_3$ and $g_5$ are independent of each other.
    \item Introduce strongly possible versions of multivalued dependencies and cross joins
    \item Compare spMVDs with NMVDs introduced by Lien \cite{lien1979multivalued}.
    \item Analyse the complexity of checking the validity of a given spCJ, and also determining the approximation measures of a given spCJ. 
\end{itemize}
Possible worlds were studied by many sets of authors, such as \cite{kohler2015possible,de2004possible,zimanyi1997imperfect}. Generally, a \nul can be replaced by any value from a domain that can be infinite, then an infinite number of worlds can be considered. While strongly possible worlds considers only values allowed by the original table given to replace with \nul s.
The paper is organized as follows. Section \ref{defs} gives the basic definitions and notations. Section \ref{rltd-wrk} discusses some related work. Sections \ref{spkeySec} and \ref{spfdSec} provide the approximation measures for spKeys and spFDs respectively. Strongly possible versions of tuple generating constraints are introduced in Section~\ref{sec:tuple-generating}. Their approximation measures are analysed, furthermore complexity questions are treated in this section.  The conclusions and some suggested future research directions are discussed in Section \ref{cnclosnSec}.

\section{Basic Definitions}\label{defs}
Let $ R = \{ A_{1},A_{2},\ldots A_{n}\} $ be a relation schema. The set of all the possible values for each attribute $ A_i \in R $ is called the domain of $A_i$ and denoted as $ D_{i}$ = $dom(A_{i})$ for $i$ = 1,2,$\ldots n$. Then, for $X\subseteq R$, then $D_X = \prod\limits_{\forall A_i \in K} D_i$.

An instance $T$ = ($t_{1}$,$t_{2}, \ldots t_{s}$) over $R$ is a list of tuples such that each tuple is a function $t :  R \rightarrow \bigcup_{A_i\in R} dom(A_i)$ and $t[A_i] \in dom(A_i)$ for all $A_i$ in $R$. By taking a list of tuples we use the \emph{bag semantics} that allows several occurrences of the same tuple. Usage of the bag semantics is justified by that SQL allows multiple occurrences of tuples. Of course, the order of the tuples in an instance is irrelevant, so mathematically speaking we consider a \emph{multiset of tuples} as an instance.
For a tuple $t_{r} \in T$ and $X \subset R$, let $t_{r}[X]$ be the restriction of $t_{r}$  to $X$.

It is assumed that $\bot$ is an element of each attribute's domain that denotes missing information.
$t_{r}$ is called $V$-total for a set $V$ of attributes  if $\forall A\in V$, $t_r[A]\neq\bot$. Also, $t_{r}$ is a total tuple if it is $R$-total.   $t_{1}$ and $t_{2}$ are \emph{weakly similar} on $X \subseteq R$ denoted as $t_{1}[X] \sim_{w} t_{2}[X]$ defined by K\"ohler et.al. \cite{kohler2016possible} if 

\[ \forall A \in X \quad (t_{1}[A] = t_{2}[A] \textrm{ or } t_{1}[A] =\bot \textrm{ or } t_{2}[A] = \bot). \]

Furthermore, $t_{1}$ and $t_{2}$ are \emph{strongly similar} on $X \subseteq R$ denoted by $t_{1}[X] \sim_{s} t_{2}[X]$ if

$$  \forall A \in X \quad  (t_{1}[A] = t_{2}[A] \neq \bot).$$
For the sake of convenience we write $t_{1} \sim_{w} t_{2}$ if $t_{1}$ and $t_{2}$ are weakly similar on $R$ and use the same convenience for strong similarity.
Let $T= (t_{1},t_{2}, \ldots t_{s})$ be a table instance over $R$. Then, $T'= (t'_{1},t'_{2}, \ldots t'_{s})$ is a \emph{possible world} of $T$, if $t_{i} \sim_{w} t'_{i}$ for all $i=1,2,\ldots s$ and $T'$ is completely \nul -free. That is, we replace the occurrences of $\bot$ with a value from the domain $ D_{i} $ different from $\bot$ for all tuples and all attributes. A active domain of an attribute is the set of all the distinct values shown under the attribute except the \nul. Note that this was called the \emph{visible domain} of the attribute in papers \cite{alattar2019keys,alattar2020functional,alattar2020strongly,al2022strongly}. 

\begin{definition}\label{vd-def} The \emph{active domain} of an attribute $ A_i $ ($ VD^T_{i} $) is the set of all distinct values except $\bot$ that are already used by tuples in $ T $:
$$ VD^T_{i} = \{t[A_{i}] : t \in T\} \setminus \{\bot \} \textrm{ for } A_i \in R.$$ 
\end{definition}To simplify notation, we omit the upper index $T$ if it is clear from the context what instance is considered.

Then the $ VD_{1} $ in Table~\ref{fig:cai1} is \{Mathematics, Datamining\}. The term active domain refers to the data that already exist in a given dataset. For example, if we have a dataset with no information about the definitions of the attributes' domains, then we use the data itself to define their own structure and domains. This may provide more realistic results when extracting the relationship between data so it is more reliable to consider only what information we have in a given dataset. 

While a possible world is obtained by using the domain values instead of the occurrence of \nul, a strongly possible world is obtained by using the active domain values. 

\begin{definition}\label{def:spWorld} A possible world $T^\prime$ of $T$ is called a \emph{strongly possible world (spWorld)} if  $t'[A_i]\in VD^T_i$ for all $t'\in T'$ and $A_i\in R$. 
\end{definition}
Note that $VD^T_i=\emptyset$ might happen in the degenerate case of an attribute having only \nul\ values. In order to be able to define strongly possible world for any table, a special symbol $ssymb$ is added to $VD^T_i$ in that degenerate case. The symbol $ssymb$ is assumed not being contained in any domain of any attribute.
The concept of \emph{strongly possible world} was introduced in \cite{alattar2019keys}. A strongly possible worlds allow us to define \emph{strongly possible keys (spKeys)} and \emph{strongly possible functional dependencies (spFDs)}.
\begin{definition}\label{spfd-spkey_def} A strongly possible functional dependency, in notation $X\spfd Y$, holds in table $T$ over schema $R$ if there exists a strongly possible world $T'$ of $T$ such that $T'\models X\fd Y$. That is, for any $t'_1,t'_2\in T'$ $t'_1[X]=t'_2[X]$ implies $t'_1[Y]=t'_2[Y]$.
  The set of attributes $X$ is a strongly possible key, if there exists a strongly possible world $T'$ of $T$ such that $X$ is a key in $T'$, in notation $\spk{X}$. That is, for any $t'_1,t'_2\in T'$ $t'_1[X]=t'_2[X]$ implies $t_1'=t_2'$. 
\end{definition}
Note that this is not equivalent with spFD $X\spfd R$, since we use the bag semantics. For example, \{Course Name, Year\} is a strongly possible key of Table~\ref{fig:cai1} as the strongly possible world in Table~\ref{fig:cai2} shows it.

\begin{table}[ht]
    \centering
\begin{tabular}{ l c c c c }
		\hline
		\textbf{Course Name} & \textbf{Year} & \textbf{Lecturer}  & \textbf{Credits} & \textbf{Semester}    \\  \hline
		Mathematics & 2019 & $\bot$ & 5  & 1 \\ 
		Datamining & 2018 & Sarah & 7 & $\bot$  \\ 
		$\bot$ & 2019 & Sarah & $\bot$ & 2 \\ \hline
	\end{tabular}
    \caption{Incomplete Dataset}
    \label{fig:cai1}
\end{table}
\begin{table}[ht]
    \centering
    \begin{tabular}{ l c c c c }
		\hline
		\textbf{Course Name} & \textbf{Year} & \textbf{Lecturer}  & \textbf{Credits}  & \textbf{Semester}  \\  \hline
		Mathematics & 2019 & Sarah & 5 & 1  \\ 
		Datamining & 2018 & Sarah & 7 & 2 \\ 
		Datamining & 2019 & Sarah & 7 & 2 \\ \hline
	\end{tabular}
	\caption{Complete Dataset}
    \label{fig:cai2}
\end{table}
Following \cite{hannula2022discovery} strongly possible cross joins are defined as
\begin{definition}\label{def:spCJ}
    A \emph{cross join (CJ)} over relation schema $R$ is an expression
$X \times Y$ where $X$ and $Y$ are subsets of $R$. A table $T$ over
$R$ is said to satisfy the CJ $X \times Y$ over $R$ if and only if for
all $t_1, t_2 \in T$ there is some $t\in r$ such that $t[X] = t_1[X]$ and
$t[Y ] = t_2[Y ]$. If $T$ does not satisfy $X \times Y$, then we also say that $r$ \emph{violates} $X \times Y$. A \emph{strongly possible cross join (spCJ)} over schema $R$ is an expression $X \times_{sp} Y$ where $X$ and $Y$ are subsets of $R$. A table $T$ over
$R$ is said to satisfy the spCJ $X \times_{sp} Y$ over $R$ if and only if there exists a strongly possible world $T'$ of $T$ such that $T'\models X \times Y$.
\end{definition}
Similarly, \emph{strongly possible multivalued dependencies} are introduced as follows.
\begin{definition}
A \emph{strongly possible multivalued dependency (spMVD)} over schema $R$ is an expression $X \spmvd Y$. The spMVD $X \spmvd Y$ is satisfied in $T$ if and only if there exists a strongly possible world $T^{\prime}$ of $T$ such that $T^{\prime} \models X\mvd Y$.
\end{definition}

If $T=\{t_1,t_2,\ldots ,t_p\}$ and  $T'=\{t'_1,t'_2,\ldots ,t'_p\}$ is an spWorld of it with $t_i\sim_w t_i'$, then $t_i'$ is called an \emph{sp-extension} or in short an \emph{extension} of $t_i$. Let $X\subseteq R$ be a set of attributes and let $t_i\sim_w t_i'$ such that for each $A\in R\colon t_i'[A]\in VD(A)$, then $t_i'[X]$ is an \emph{strongly possible extension of $t_i$ on $X$ (sp-extension)}

\section{Related Work}\label{rltd-wrk}
\subsubsection{Approximation measures}
Giannella et al. \cite{giannella2004approximation} measure the approximate degree of functional dependencies. They developed the IFD approximation measure and compared it with the other two measures: $g_3$ (minimum number of tuples need to be removed so that the dependency holds) and $\tau$ (the probability of a correct guess of an FD satisfaction) introduced in \cite{kivinen1995approximate} and \cite{goodman1979measures} respectively. They developed analytical bounds on the measure differences and compared these measures analysis on five datasets. The authors show that when measures are meant to define the knowledge degree of $X$ determines $Y$ (prediction or classification), then $IFD$ and $\tau$ measures are more appropriate than $g_3$. On the other hand, when measures are meant to define the number of "violating" tuples in an FD, then, $g_3$ measure is more appropriate than $IFD$ and $\tau$.
This paper extends the earlier  work of \cite{alattar2020strongly} that utilized the $g_3$ measure for spKeys by calculating the minimum number of tuples to be removed from a table so that an spKey holds if it is not. The same paper proposed the $g_4$ measure that is derived from $g_3$ by emphasizing the effect of each connected component in the table's corresponding bipartite graph (where vertices of the first class of the graph represent the table's tuples and the second class represent all the possible combinations of the attributes' active domains). 
In this paper, we propose a new measure  $g_5$ to approximate FDs by adding new tuples with unique values rather than deleting tuples as in $g_3$. 

Several other researchers worked on approximating FDs in the literature.  King et al. \cite{king2003discovery} provided an algorithmic method to discover functional and approximate functional dependencies in relational databases. The method provided is based upon the mathematical theory of partitions of row identification numbers from the relation, then determining non-trivial minimal dependencies from the partitions. They showed that the operations that need to be done on partitions are both simple and fast.
\subsection{Algorithmic approaches}
In \cite{varkonyi2017multi}, Varkonyi et al. introduced a structure called Sequential Indexing Tables (SIT) to detect an FD regarding the last attribute in their sequence. SIT is a fast approach so it can process large data quickly. The structure they used does not scale efficiently with the number of the attributes and the sizes of their domains, however. Other methods, such as TANE and FastFD face the same problem \cite{papenbrock2015functional}. TANE was introduced by Huhtala \cite{huhtala1999tane} to discover functional and approximate dependencies by taking into consideration partitions and deriving valid dependencies from these partitions in a breadth-first or level-wise manner. 

Bra, P. De, and Jan Paredaens gave a new decomposition theory for functional dependencies in \cite{bra1983conditional}. They break up a relation into two subrelations whose union is the given relation and a functional dependency that holds in one subrelation is not in the other. 

In \cite{tusor2020memory}, Tusor et al. presented the Parallelized Sequential Indexing Tables method that is memory-efficient for large datasets to find exact and approximate functional dependencies. Their method uses the same principle of Sequential Indexing Tables in storing data, but their training approach and operation are different. 

Pyro is an algorithm to discover all approximate FDs in a dataset presented by Kruse \cite{kruse2018efficient}. Pyro verifies samples of agree sets and prunes the search spaces with the discovered FDs. On the other hand, based on the concept of "agree sets", Lopes et al. \cite{lopes2002functional} developed an algorithm to find a minimum cover of a set of FDs for a given table by applying the so-called "Luxenburger basis" to develop a basis of the set of  approximate FDs in the table.

Simovici et al. \cite{simovici2002impurity} provide an algorithm to find purity dependencies such that, for a fixed right-hand side ($Y$), the algorithm applies a level-wise search on the left-hand sides ($X$) so that $X\fd Y$ has a purity measure below a user-defined threshold.
Other algorithms were proposed in \cite{kantola1992discovering,lopes2000efficient} to discover all FDs that hold in a given table by searching through the lattice of subsets of attributes.

In \cite{wijsen2019foundations}, Jef Wijsen summarizes and discusses some theoretical developments and concepts in Consistent query answering CQA (when a user queries a database that is inconsistent with respect to a set of constraints). Database repairing was modeled by an acyclic binary relation $\leq_{db}$ on the set of consistent database instances, where $r_1$ $\leq_{db}$ $r_2$ means that $r_1$ is at least as close to $db$ as $r_2$. One possible distance is the number of tuples to be added and/or removed. In addition to that, Bertossi studied the main concepts of database repairs and CQA in \cite{bertossi2019database}, and emphasis on tracing back the origin, motivation, and early developments.
J. Biskup and L. Wiese present and analyze an algorithm called preCQE that is able to correctly compute a solution instance, for a given original database instance, that obeys the formal properties of inference-proofness and distortion minimality of a set of appropriately formed constraints in \cite{biskup2011sound}. 

Hannula et. al. \cite{hannula2022discovery} showed that the problem deciding whether there is a cross join that holds on a gven table is W[3] complete. The present authors \cite{alattar2020functional} showed that checking one spKey can be done in polynomial time, while checking a system of spKeys, as well as checking an spFD is NP complete. In a forthcoming paper \cite{alatar2024} they also treat the complexity of determining the approximation measures of spKeys and spFDs.
\subsection{Multivalued dependencies}
Multivalued dependencies are a generalization of functional dependencies and they extend the understanding of logical database design. Multivalued dependencies lead to the fourth normal form for relational databases, where a relation schema is in the fourth normal form if all functional and multivalued dependencies are the result of keys \cite{fagin1977multivalued}.
The satisfaction of the functional dependency $X \fd Y$ is a sufficient condition for $r$ to be the lossless join of its projections $r[X Y]$ and $r[X(R \setminus Y)]$, but not a necessary condition. The satisfaction of multivalued dependency $X \mvd Y$ provides a sufficient and necessary condition for $r$ to be the lossless join of its projections $r[X Y]$ and $r[X(R \setminus Y)]$ without loss of information \cite{codd1972further}.
Informally, a relation $r$ over the relation schema $R$ satisfies the MVD $X \mvd Y$ whenever the value on $X$ determines the set of values on $Y$ independently
of the set of values on $R\setminus Y$.

The multivalued dependencies show different behaviour in the presence of missing information in the attributes of the table, because missing information changes the implication properties between attributes. For example, if the two tuples $t_1$ and $t_2$ have the values $1$ and $\bot$ respectively on the attribute $X$, they do not agree on $X$ and they do not disagree at the same time. It is possible that $t_2[X] = 1$, and this value is not yet inserted. Some research and studies discuss and analyze the problem of incomplete databases for multivalued dependencies. 

\cite{HartLink2012,lien1979multivalued,lien1982equivalence} introduced a generalized version of multivalued dependencies for databases with missing information. 
Lien \cite{lien1979multivalued} introduced a new class of multivalued dependencies in the context of null values have been, called \emph{multivalued dependency with nulls} (NMVD). The NMVD N$X\mvd Y$ holds in $R(X,Y,Z)$, where $Z$ is disjoint from $Y$, if and only if, whenever $X$ null-free ($X$-total) tuples $(x,y,z)$ and $(x,y',z')$ are tuples in $R$, so are tuples $(x,y',z)$ and $(x,y,z')$. The data value $x$ is for attributes $X$, and similarly, $y$ (or $y'$) and $z$(or $z'$) correspond to $Y$ and $Z$ respectively. ($R$ contains no two distinct tuples $t$ and $t'$ such that $t$ subsumes $t'$). In \cite{lien1982equivalence}, Lien defined the NMVD over $R$ as a statement $X\mvd Y$ and considers that the left-hand side of the dependency is $X$-total. According to this, MVD $X\mvd Y$ is satisﬁed by a relation $R$, denoted by $R \models X\mvd Y$, if and only if for every two tuples $t_1, t_2$ in the relation, the following is true: if $t_1$ and $t_2$ are $X$-total and $t_1[X] = t_2[X]$, then there is some tuples $t$ such that $t[XY] = t_1[XY]$ and $t[X(R \setminus Y)] = t_2[X(R \setminus Y)]$. Lien presented valid inference rules. It was also shown that the four basic inference rules, complementation, reflexivity, augmentation, and union, are complete and therefore sufficient for studying the formal properties of the revised multivalued dependencies.

S. Link \cite{HartLink2012} defined the \emph{weak multivalued dependency} (WMVDs) over R as a statement $\Diamond X\mvd Y$, where $XY \subseteq R$. A relation $r$ over $R$ is said to satisfy the WMVD $\Diamond X\mvd Y$ over $R$ if there is some possible world $s$ of $r$ such that for all $t_1, t_2 \in s$ the following holds: if $t_1[X] = t_2[X]$, then there is some $t \in s$ such that $t[XY] = t_1[XY]$ and $t[X(R\setminus Y)] = t_2[X(R\setminus Y)]$.
Our proposed class of multivalued dependency uses the strongly possible worlds instead of the possible worlds.

\section{SPKey Approximation}\label{spkeySec}

{Investigations of strongly possible worlds can be extended for other types of dependencies, like multivalued dependencies and cross joins. 
In \cite{alattar2020strongly}, the authors studied strongly possible keys, and the main motivation is to uniquely identify tuples in incomplete tables, if it is possible, by using the already shown values only to fill up the occurrences of \nul s. Consider the relational schema $R=$ and $K\subseteq R$. Furthermore, let $T$ be an instance over $R$ with \nul s. 
Let $T'$ be the set of total tuples  $T'=\{t'\in\Pi_{i=1}^b VD^T_i\colon \exists t\in T \text{ such that } t[K]\sim_w t'[K]\}$, furthermore let $G=(T,T';E)$ be the bipartite graph, called the \emph{$K$-extension graph of $T$}, defined by $\{t,t'\}\in E\iff t[K]\sim_w t'[K]$.
Finding a matching of $G$ that covers all the tuples in $T$ (if exists) provides the set of tuples in $T^{\prime}$ to replace  the incomplete tuples in $T$ with, to verify that $K$ is an spKey.
A polynomial-time algorithm was given in \cite{alattar2020functional} to find such matching. It is a non-trivial application of the well-known matching algorithms, as $|T'|$ is usually an exponential function of the  size of the input table $T$.

The Approximate Strongly Possible Key (ASP Key) was defined in \cite{alattar2020strongly} as follows.
\begin{definition}\label{aspkey-def}
Attribute set $K$ is an approximate strongly possible key of ratio $a$ in table $T$, in notation $asp_a^- \left\langle K \right\rangle$, if there exists a subset  $S$ of the tuples $T$ such that $T\setminus S$ satisfies $sp \left\langle K \right\rangle$, and $|S|/|T|\le a$. The minimum $a$ such that $asp_a^- \left\langle K \right\rangle$ holds is denoted by $g_3(K)$. 
\end{definition} 
The measure $g_3(K)$ represents the approximation which is the ratio of the number of tuples needed to be removed over the total number of tuples so that $sp \left\langle K \right\rangle$ holds. 
The measure $g_3(K)$ has a value between $0$ and $1$, and it is exactly $0$ when $ sp \left\langle K \right\rangle $ holds in $T$, which means we don't need to remove any tuples. For this, we used the $g_3$ measure introduced in \cite{kivinen1995approximate}, to determine the degree to which $ASP$ key is approximate. For example, the $g_3$ measure of $\spk{X}$ on Table \ref{fig:spkey_main} is 0.5, as we are required to remove two out of four tuples to satisfy the key constraint as shown in Table \ref{fig:spkey_rmv}. 

It was shown in \cite{alattar2020strongly} that the $g_3$ approximation measure for strongly possible keys satisfies
$$ g_3(K) = \frac{|T| -  \nu(G)}{|T|}. $$
where $\nu(G)$ denotes the maximum matching size in the $K$-extension graph $G$. 
The smaller value of $g_3(K)$, the closer $K$ is to being an spKey.

For the bipartite graph $G$ defined above, let $\mathscr{C}$ be the collection of all the connected components in $G$ that satisfy the spKey,
i.e. for which there exists a  matching that covers all tuples in the set ($\forall_{C\in \mathscr{C}} $ $ \nexists X \subseteq C \cap T$ such that $|X| > N(X)$ by Hall's Theorem). 
Let $D \subseteq G$ be defined as $D = G \setminus \bigcup_{\forall C \in \mathscr{C}} C$, and let $\mathscr{C^\prime}$ be the set of connected components of $D$. Let $V_C$ denote the set of vertices in a connected component $C$.
The approximation measure of strongly possible keys may be more appropriate by considering the effect of each connected component in the bipartite graph on the matching.
We consider the effect of the components of $\mathscr{C}$ to get doubled in the approximation measure, as these components represent that part of the data that do not require tuple removal. 
So a derived version of the $g_3$ measure was proposed and named $g_4$ considering these components' effects,
$$ g_4(K) = \frac{|T| - (\sum_{C\in \mathscr{C}} (|V_C|) + \sum_{ C^\prime \in \mathscr{C^\prime}}\nu(C^\prime))}{|T|+ \sum_{C \in \mathscr{C}} |V_C|}.$$

Furthermore, it was proved that for a set of attributes $K$ in any table, we have either $g_3(K)= g_4 (K)$ or  $1 <g_3(K) / g_4 (K) < 2 $. Moreover, there exist tables of an arbitrarily large number of tuples with $g_3(K) / g_4(K)=\frac{p}{q}$ for any rational number $1\le \frac{p}{q} <2$.

In this paper, we extend our investigation on approximating spKeys by considering adding new tuples instead of removing them to satisfy an spKey if possible. Removing a non-total tuple $t_1$ means that there exist another total and/or non-total tuple(s) that share the same strongly possible extension with $t_2$. The following proposition shows that we can always remove only non-total tuples if the total part of the table satisfies the key.  
\begin{proposition}\label{prop:removenontotal}
Let $T$ be an instance over schema $R$ and let $K\subseteq R$. If the $K$-total part of the table $T$ satisfies the key $sp \left\langle K \right\rangle$, then there exists a minimum set of tuples $U$ to be removed that are all non-$K$-total so that $T\setminus U$ satisfies $sp \left\langle K \right\rangle.$
\end{proposition}

\begin{proof} Observe that a minimum set of tuples to be removed is $T\setminus X$ for a subset $X$ of the set of vertices (tuples) covered by a particular maximum matching of the $K$-extension graph.
Let $M$ be a maximum matching, and assume that $t_1$ is total and not covered by ${M}$. Then, the unique neighbour $t_1^{\prime}$ of $t_1$ in $T^{\prime}$ is covered by an edge $(t_2, t_1^{\prime})$ of $\mathcal{M}$. Then $t_2$ is non-total since the $K$-total part satisfies $sp \left\langle K \right\rangle$, so we replace the edge $(t_2, t^{\prime})$ by the edge $(t_1, t^{\prime})$ to get matching  $M_1$ of size $|M_1 |= |M|$, and $M_1$ covers one more total tuple. Repeat this until all total tuples are covered.
\end{proof}

\subsection{ Measure $g_5$ for spKeys}

The $g_3$ approximation measure for spKeys was introduced in \cite{alattar2020strongly}. In this section, we introduce a new approximation measure for spKeys. As we consider the active domain to be the source of the values to replace each null with, adding a new tuple to the table may increase the number of the values in the active domain of an attribute. for example, consider Table \ref{fig:spkey_main}, the active domain of the attribute $X_1$ is $\{2\}$ and it changed to $\{2, 3\}$ after adding a tuple with new values as shown in Table \ref{fig:spkey_add}.

\begin{minipage}[c]{0.33\textwidth}
	\centering
        \begin{tabular}{ll}
        			\hline
\multicolumn{2}{c}{\textbf{X}} \\
$X_1$              & $X_2$          
\\\hline
$\bot$ & 1 \\
2 & $\bot$ \\
2 & $\bot$ \\
2 & 2
\\ \hline
\end{tabular}
        		\captionof{table}{Incomplete Table\\ to measure $\spk{X}$}
                \label{fig:spkey_main}
            	\end{minipage}	\begin{minipage}[c]{0.33\textwidth}
            	\centering
            		\begin{tabular}{ll}
        			\hline
\multicolumn{2}{c}{\textbf{X}} \\
$X_1$              & $X_2$          
\\\hline
$\bot$ & 1 \\
2 & 2 \\ \hline
\end{tabular}
        		\captionof{table}{The table after\\ removing ($asp_a^-\left\langle X \right\rangle$)}
                \label{fig:spkey_rmv}
            	\end{minipage}	\begin{minipage}[c]{0.33\textwidth}
            	
            	\centering
            		\begin{tabular}{ll}
        			\hline
\multicolumn{2}{c}{\textbf{X}} \\
$X_1$              & $X_2$          
\\\hline
$\bot$ & 1 \\
2 & $\bot$ \\
2 & $\bot$ \\
2 & 2 \\
3 & 3 
\\ \hline
\end{tabular}
        		\captionof{table}{The table\\ after adding ($asp_b^+\left\langle X \right\rangle$)}
                \label{fig:spkey_add}
\end{minipage}
\vspace{0.9mm}

In the following definition, we define the $g_5$ measure as the ratio of the minimum number of tuples that need to be added over the total number of tuples to have the spKey satisfied.

\begin{definition}
Attribute set $K$ is an add-approximate strongly possible key of ratio $b$ in table $T$, in notation $asp_b^+\left\langle K \right\rangle$, if there exists a set of tuples $S$ such that the table $TS$ satisfies $sp \left\langle K \right\rangle$, and $|S|/|T|\le b$. The minimum $b$ such that $asp_b^+\left\langle K \right\rangle$ holds is denoted by $g_5(K)$. 
\end{definition} 
The measure $g_5(K)$ represents the approximation which is the ratio of the number of tuples needed to be added over the total number of tuples so that $sp \left\langle K \right\rangle$ holds. 
The value of the measure $g_3(K)$ ranges between $0$ and $1$, and it is exactly $0$ when $ sp \left\langle K \right\rangle $ holds in $T$, which means we do not have to add any tuple.  For example, the $g_5$ measure of $\spk{X}$ on Table \ref{fig:spkey_main} is 0.25, as it is enough to add one tuple to satisfy the key constraint as shown in Table \ref{fig:spkey_add}.

Let $T$ be a table and $U \subseteq T$ be the set of the tuples that we need to remove so that the spKey holds in $T$, i.e, we need to remove $|U|$ tuples, while by adding a tuple with new values, we may make more than one of the tuples in $U$ satisfy the spKey using the new added values for their \nul s. In other words, we may need to add a fewer number of tuples than the number of tuples we need to remove to satisfy an spKey in the same given table. For example, Table \ref{fig:spkey_main} requires removing two tuples to satisfy $sp \left\langle X \right\rangle$, while adding one tuple is enough.

On the other hand, one may think about mixed modification of both adding and deleting tuples for Keys approximation, by finding the minimum number of tuples needs to be either added or removed. If first the additions are performed, then after that by Proposition \ref{prop:removenontotal}, it is always true that we can remove only non-total tuples; then, instead of any tuple removal, we may add a new tuple with distinct values. Therefore, mixed modification in that way would not change the approximation measure, as it is always equivalent to tuples addition only. However, if the order of removals and additions count, then it is a topic of further research whether the removals can be substituted by additions.

The values of the two measures, $g_3$ and $g_5$, range between $0$ and $1$, and they are both equal to 0 if the spKey holds (we do not have to add or remove any tuples). Proposition \ref{g3_l_g5} proves that the value of $g_3$ measure is always larger than or equal to the value of $g_5$ measure.   

\begin{proposition}\label{g3_l_g5}
For any $K\subseteq R$ with $|K|\geq 2$, we have $g_3(K)\geq g_5(K)$.
\end{proposition}

\begin{proof}
Indeed, we can always remove non-total tuples for $g_3$ by Proposition~\ref{prop:removenontotal}.  Let the tuples to be removed  be $U=\{t_1,t_2,\ldots t_u\}$. Assume that $T^*$ is an spWorld of $T\setminus U$, which certifies that $T\setminus U\models sp \left\langle K \right\rangle$ For each tuple $t_i\in U$, we add tuple $t_i'=(z_i,z_i,\ldots ,z_i)$ where $z_i$ is a value that does not occur in any other tuple originally of $T$ or added. The purpose of adding $t_i'$ is twofold. First it is intended to introduce a completely new active domain value for each attribute. Second, their special structure ensures that they will never agree with any other tuple in the spWorld constructed below for the extended instance. Let $t_i"$ be a tuple such that exactly one \nul\ in $K$ of $t_i$ is replaced by $z_i$, any other \nul s of $t_i$ are imputed by values from the original active domain of the attributes. It is not hard to see that tuples in $T^*\cup\{t_1',t_2'\ldots ,t_u'\}\cup\{t_1",t_2"\ldots ,t_u"\}$ are pairwise distinct on $K$.
\end{proof}
According to Proposition~\ref{g3_l_g5} we have $0\le g_3(K)-g_5(K)<1$ and the difference is a rational number. What is not immediate is that for any rational number $0\le\frac{p}{q}<1$ there exist a table $T$ and $K\subseteq R$ such that $g_3(K)-g_5(K)=\frac{p}{q}$ in table $T$. 
\begin{proposition}\label{prop:spklarge}
Let $0\le\frac{p}{q}<1$ be a rational number. Then there exists a table $T$ with an arbitrarily large number of rows and $K\subseteq R$ such that $g_3(K)-g_5(K)=\frac{p}{q}$ in table $T$. 
\end{proposition}
\begin{proof}
We may assume without loss of generality that $K=R$, since $T'\models sp \left\langle K \right\rangle$ if and only if we can make the tuples pairwise distinct on $K$ by imputing values from the active domains, that is values in $R\setminus K$ are irrelevant. Let $T$ be the following $q\times(p+2)$ table (with $x=q-p-1$).
\begin{equation}\label{eq:spklarge}
  T=\begin{array}{rl}
  \left.  \begin{array}{p{7pt}p{7pt}p{7pt}cc}
     1   & 1&1&\ldots &1 \\
     1   & 1&1&\ldots &2\\
     & &\vdots& & \\
     1   & 1&1&\ldots &x
   \end{array} \right\}   & q-p-1 \\
   \left. \begin{array}{ccccc}
     \bot    &1&\ldots&1&1  \\
      1   & \bot&\ldots&1&1  \\
      & &\ddots& & \\
      1&1&\ldots&\bot&1
    \end{array} \right\}  & p+1
  \end{array}
\end{equation}
Since the active domain of the first $p+1$ attributes is only $\{1\}$, we have to remove $p+1$ rows so $g_3(K)=\frac{p+1}{q}$. On the other hand it is enough to add one new row $(2,2,\ldots ,2,q-p)$ so $g_5(K)=\frac{1}{q}$. Since $\frac{p}{q}=\frac{cp}{cq}$ for any positive integer $c$, the number of rows in the table could be arbitrarily large.
\end{proof}
The tables constructed in the proof of Proposition~\ref{prop:spklarge} have an arbitrarily large number of rows, however, the price for this is that the number of columns is also not bounded. The question arises naturally whether there are tables with a fixed number of attributes but with an arbitrarily large number of rows  that  satisfy $g_3(K)-g_5(K)=\frac{p}{q}$ for any rational number $0\le\frac{p}{q}<1$? The following theorem answers this problem.
\begin{theorem}\label{thm:spkbounded}
Let $0\le\frac{p}{q}<1$ be a rational number. Then there exist tables over schema $\{A_1,A_2\}$ with arbitrarily large number of rows, such that $g_3(\{A_1,A_2\})-g_5(\{A_1,A_2\})=\frac{p}{q}$.
\end{theorem}
\begin{proof}
The proof is divided into three cases according to whether $\frac{p}{q}<\frac{1}{2}$, $\frac{p}{q}=\frac{1}{2}$ or $\frac{p}{q}>\frac{1}{2}$. In each case, the number of rows of the table will be an increasing function of $q$ and one just has to note that $q$ can be chosen arbitrarily large without changing the value of the fraction $\frac{p}{q}$.
\paragraph{Case $\frac{p}{q}<\frac{1}{2}$} Let $T_{<.5}$ be defined as 
\begin{equation}
   T_{<.5}=\begin{array}{rl}
 q-p-1  & \left\{\begin{array}{ccl}
         1  &\text{\ } &1 \\
         1  & &2 \\
         \vdots&&\vdots\\
         1&&q-p-1
      \end{array}\hfil \right.\\
     p+1  & \left\{\begin{array}{ccl}
         \bot&\text{\ }&\hfil\bot\\
         \bot & & \hfil \bot \\
         \vdots&&\hfil \vdots\\
         \bot  && \hfil \bot
      \end{array} \right.
           \end{array} 
\end{equation}
Clearly, $g_3(K)=\frac{p+1}{q}$, as all the tuples with \nul s have to be removed. On the other hand, if tuple $(2,q-p)$ is added, then the total number of active domain combinations is $2\cdot(q-p)$, out of which $q-p$ is used up in the table, so there are $q-p$ possible pairwise distinct tuples to replace the \nul s. Since $\frac{p}{q}<\frac{1}{2}$, we have that $q-p\ge p+1$ so all the tuples in the $q+1$-rowed table can be made pairwise distinct. Thus, $g_3(K)-g_5(K)=\frac{p+1}{q}-\frac{1}{q}$. 
\paragraph{Case $\frac{p}{q}=\frac{1}{2}$} Let $T_{=0.5}$ be defined as 
\begin{equation}
   T_{=.5}=\begin{array}{rl}
      q-p-2  & \left\{\begin{array}{ccl}
         1  &\text{\ }& 1 \\
         1  && 2 \\
         \vdots&&\vdots\\
         1&&q-p-2
      \end{array} \right.\\
      p+2 & \left\{\begin{array}{ccl}
         \bot &\text{\ } & \hfil \bot \\
         \bot  && \hfil \bot\\
         \vdots&&\hfil\vdots\\
         \bot  && \hfil \bot
      \end{array}\right.
   \end{array} 
 \end{equation}

Table $T_{=.5}$  contains all possible combinations of the active domain values, so we have to remove every tuple containing \nul s, so $g_3(K)=\frac{p+2}{q}$. On the other hand, if we add just one new tuple (say $(2,q-p-1)$), then the largest number of active domain combinations is $2\cdot(q-p-1)$ that can be achieved. There are already $q-p-1$ pairwise distinct total tuples in the extended table, so only $q-p-1<p+2$ would be available to replace the \nul s. On the other hand, adding two new tuples, $(2,q-p-1)$ and $(3,q-p)$ creates a pool of $3\cdot(q-p)$ combinations of active domains, which is more than $(q-p-1)+p+2$ that is needed.
\paragraph{Case $\frac{p}{q}>\frac{1}{2}$} Table $T$ is defined similarly to the previous cases, but we need more careful analysis of the numbers.
\begin{equation}\label{eq:spkbd>.5}
   T=\begin{array}{rl}
      b & \left\{\begin{array}{cc}
         1  &\,\, 1 \\
         1  & \,\,2 \\
         \vdots&\,\,\vdots\\
         1&\,\,b
      \end{array} \right.\\
      x & \left\{\begin{array}{cc}
         \bot  & \bot \\
         \bot  & \bot\\
         \vdots&\vdots\\
         \bot  & \bot
      \end{array}\right.
   \end{array} 
\end{equation}
Clearly, $g_3(K)=\frac{x}{x+b}$. Let us assume that $y$ tuples are needed to be added. The maximum number of active domain combinations is $(y+1)(y+b)$ obtained by adding tuples $(2,b+1),(3,b+2),\ldots ,(y+1,y+b)$. This is enough to replace all tuples with \nul s if
\begin{equation}\label{eq:y-lower}
(y+1)(y+b)\ge x+y+b.    
\end{equation}
On the other hand, $y-1$  added tuples are not enough, so 
\begin{equation}\label{eq:y-upper}
y(y-1+b)<x+y-1+b.  
\end{equation}
Since the total number of active domain combinations must be less than the tuples in the extended table. We have $\frac{p}{q}=g_3(K)-g_5(K)=\frac{x-y}{x+b}$ that is for some positive integer $c$ we must have $cp=x-y$ and $cq=x+b$ if $gcd(p,q)=1$. This can be rewritten as 
\begin{equation}\label{eq:xyb-cpq}
    y=x-cp \; ;\quad y+b=c(q-p) \; ;\quad b=cq-x \; ;\quad x+y+b=y+cq .
\end{equation}
Using (\ref{eq:xyb-cpq}) we obtain that (\ref{eq:y-lower}) is equivalent with 
\begin{equation}\label{eq:ylowceil}
    y\ge\frac{cp}{c(q-p)-1}.
\end{equation}
If $c$ is large enough then $\lceil \frac{cp}{c(q-p)-1}\rceil=\lceil\frac{p}{q-p}\rceil$ so if $y=\lceil\frac{p}{q-p}\rceil$ is chosen then (\ref{eq:ylowceil}) and consequently (\ref{eq:y-lower}) holds. On the other hand, (\ref{eq:y-upper}) is equivalent to 
\begin{equation}\label{eq:y-upp-pq}
    y<\frac{cq-1}{c(q-p)-2}.
\end{equation}
The right hand side of (\ref{eq:y-upp-pq}) tends to $\frac{q}{q-p}$ as $c$ tends to infinity. Thus, for large enough $c$ we have $\lfloor \frac{cq-1}{c(q-p)-2}\rfloor=\lfloor \frac{q}{q-p}\rfloor$. Thus, if 
\begin{equation}\label{eq:pq-pqq-p}
y=\lceil\frac{p}{q-p}\rceil\le\lfloor \frac{q}{q-p}\rfloor
\end{equation}
and $\frac{q}{q-p}$ is not an integer, then both (\ref{eq:y-lower}) and (\ref{eq:y-upper}) are satisfied for large enough $c$. Observe that $\frac{p}{q-p}+1=\frac{q}{q-p}$, thus (\ref{eq:pq-pqq-p}) always holds. Also, if $\frac{q}{q-p}$ is indeed an integer, then we have strict inequality in (\ref{eq:pq-pqq-p}) that implies (\ref{eq:y-upp-pq}) and consequently (\ref{eq:y-upper}). 
\end{proof}

\section{spFD Approximation}\label{spfdSec}

In this section, we measure to which extent a table satisfies a Strongly Possible Functional Dependency (spFD) $X\spfd Y$  if $T \not\models X\spfd Y$. 

Similarly to Section~\ref{spkeySec}, we assume that the $X$-total part of the table satisfies the FD $X\fd Y$, so we can always consider adding tuples. The measures $g_3$ and $g_5$ are defined analogously to the spKey case.

\begin{definition}\label{spfd-approx-rmv}
For the attribute sets $X$ and $Y$, $\sigma:X \spfd Y$ is a remove-approximate strongly possible functional dependency of ratio $a$ in a table $T$, in notation \\ $T \models \approx_a^- X \spfd Y$, if there exists a set of tuples $S$ such that the table $T\setminus S \models X \spfd Y$, and $|S|/|T|\le a$. Then, $g_3(\sigma)$ is the smallest $a$ such that $T \models \approx_a^-\sigma$ holds. 
\end{definition} 
The measure $g_3(\sigma)$ represents the approximation which is the ratio of the number of tuples needed to be removed over the total number of tuples so that $T \models X \spfd Y$ holds.
\begin{definition}\label{spfd-approx-add}
For the attribute sets $X$ and $Y$, $\sigma:X \spfd Y$ is an add-approximate strongly possible functional dependency of ratio $b$ in a table $T$, in notation $T \models \approx_b^+X \spfd Y$, if there exists a set of tuples $S$ such that the table $T\cup S \models X \spfd Y$, and $|S|/|T|\le b$. Then, $g_5(\sigma)$ is the smallest $b$  such that $T \models \approx_b^+\sigma$ holds. 
\end{definition} 
The measure $g_5(\sigma)$ represents the approximation which is the ratio of the number of tuples needed to be added over the total number of tuples so that $T \models X \spfd Y$ holds.
For example, consider Table \ref{fig:spfd_main}. We are required to remove at least 2 tuples so that $X\spfd Y$ holds, as it is easy to check that if we remove only one tuple, then $T \not\models X \spfd Y$, but on the other hand, the table obtained by removing tuples 4 and 5, shown in Table \ref{fig:spfd_rmv}  satisfies $X\spfd Y$. It is enough to add only one tuple to satisfy the dependency as the table in Table \ref{fig:spfd_add} shows.

\begin{minipage}[c]{0.33\textwidth}
	\centering
        \begin{tabular}{cc|c}
        			\hline
\multicolumn{2}{c|}{\textbf{X}} & \textbf{Y} \\
$X_1$              & $X_2$        &  
\\\hline
$\bot$ & 1      & 1 \\
2      & $\bot$ & 1 \\
2      & $\bot$ & 1 \\
2      & 1      & 2 \\
2      & 1      & 2 \\
2      & 2      & 2
\\ \hline
\end{tabular}
        		\captionof{table}{Incomplete Table \\to measure ($X\spfd Y$)⟩}
                \label{fig:spfd_main}
\end{minipage}
\begin{minipage}[c]{0.33\textwidth}
            	\centering
            		\begin{tabular}{cc|c}
        			\hline
\multicolumn{2}{c|}{\textbf{X}} & \textbf{Y} \\
$X_1$              & $X_2$        &  
\\\hline
$\bot$ & 1&1\\
2      & $\bot$ & 1 \\
2      & $\bot$ & 1 \\
2&2&2
\\ \hline
\end{tabular}
        		\captionof{table}{The table after\\ removing ($_a^-X \spfd Y$)}
                \label{fig:spfd_rmv}
            	
            	\end{minipage}
\begin{minipage}[c]{0.33\textwidth}
            	\centering
            		\begin{tabular}{cc|c}
        			\hline
\multicolumn{2}{c|}{\textbf{X}} & \textbf{Y} \\
$X_1$              & $X_2$        &  
\\\hline
$\bot$ & 1      & 1 \\
2      & $\bot$ & 1 \\
2      & $\bot$ & 1 \\
2      & 1      & 2 \\
2      & 1      & 2 \\
2      & 2      & 2 \\
3      & 3      & 3
\\ \hline
\end{tabular}
        		\captionof{table}{The table after\\ adding ($_b^+X \spfd Y$)}
                \label{fig:spfd_add}

\end{minipage}

\subsection{The Difference of g3 and g5 for spFDs}

The same table may get different approximation measure values for $g_3$ and $g_5$. For example, the $g_3$ approximation measure for Table \ref{fig:spfd_main} is 0.334 (it requires removing at least 2 tuples out of 6), while its $g_5$ approximation measure is 0.167 (it requires adding at least one tuple with new values). 

The following theorem proves that it is always true that the $g_3$ measure value of a table is greater than or equal to the $g_5$ for spFDs.

\begin{theorem} \label{g3_geq_g5}
Let $T$ be a table over schema $R$, $\sigma: X\spfd Y$ for some $X,Y\subseteq R$. Then $g_3(\sigma) \geq g_5(\sigma)$.
\end{theorem}
 The proof is much more complicated than the one in the case of spKeys, because we cannot assume that there always exists a minimum set of non-total tuples to be removed for $g_3$, as the table in Table~\ref{fig:totalremove} shows. In this table the third tuple alone forms a minimum set of tuples to be removed to satisfy the dependency and it has no NULL.

 \begin{table}
     \begin{center}
 	\begin{tabular}{cc|c}
        			\hline
\multicolumn{2}{c|}{\textbf{X}} & \textbf{Y} \\
$X_1$              & $X_2$        &  
\\\hline
1& $\bot$      & 1 \\
1     & $\bot$ & 1 \\
1      & 1 & 2 \\
1      & 1 & $\bot$ \\
1      & 2 & 3\\ \hline
\end{tabular}
\end{center}
\caption{$X$-total tuple needs to be removed}
\label{fig:totalremove}
\end{table}
 From that table, we need to remove the third row to have $X\spfd Y$ satisfied. Let us note that adding row $(3,3,3)$ gives the same result, so $g_3(X\spfd Y)=g_5(X\spfd Y)=1$. However, there exist no spWorlds that realize the $g_3$ and $g_5$ measure values and agree on those tuples that are not removed for $g_3$.

\begin{proof} \textit{of Theorem~\ref{g3_geq_g5}} Without loss of generality, we may assume that $X\cap Y=\emptyset$, because $T\models X\spfd Y\iff T\models X\setminus Y\spfd Y\setminus X$. Also, it is enough to consider attributes in $X\cup Y$.
Let $U = \{t_{1}, t_{2}, \ldots, t_{p}\}$ be a minimum set of tuples to be removed from $T$. Let $T^{\prime}$ be the spWorld of $T\setminus U$ that satisfies $X\fd Y$. Let us assume that $t_1,\ldots t_a$ are such that $t_i[X]$ is not total for $1\le i\le a$. Furthermore, let $t_{a+1}[X]=\ldots =t_{j_1}[X]$, $t_{j_1+1}[X]=\ldots =t_{j_2}[X]$, $\ldots$, $t_{j_f+1}[X]=\ldots =t_{p}[X]$ be the maximal sets of tuples that have the same total projection on $X$. We construct a collection of tuples $\{s_1,\ldots s_{a+f+1}\}$, together with an spWorld $T^*$ of $T\cup\{s_1,\ldots ,s_{a+f+1}\}$  that satisfies $X\fd Y$ as follows.

\paragraph{Case 1.} $1\le i\le a$. Let $z_i$ be a value not occurring in $T$ neither in every tuple $s_j$ constructed so far. Let  $s_i[A]=z_i$ for $\forall A\in X$ and $s_i[B]=t_i[B]$ for $B\in R\setminus X$. The corresponding sp-extensions $s_i^*, t_i^*\in T^*$ are given by setting $s_i^*[B]=t_i^*[B]=\beta$ where $\beta \in VD_B$ arbitrarily fixed if $t_i[B]=\bot$ in case  $B\in R\setminus X$, furthermore $t_i^*[A]=z_i$ if $A\in X$ and $t_i[A]=\bot$. 
\paragraph{Case 2.} $X$-total tuples. For each such set $t_{j_{g-1}+1}[X]=\ldots =t_{j_g}[X]$ ($g\in\{1,2,\ldots ,f+1\}$) we construct a tuple $s_{a+g}$. Let $v_1^g,v_2^g,\ldots v_{k_g}^g\in T\setminus U$ be the tuples whose sp-extension ${v_j^g}'$ in $T'$ satisfies ${v_j^g}'[X]=t_{j_g}[X]$ for $1\le j\le k_g$. 
Let $v_1^g,v_2^g,\ldots v_{\ell}^g$ be those that are also $X$-total. Since the $X$-total part of the table satisfies $X\spfd Y$, $t_{j_{g-1}+1},\ldots t_{j_g},v_1^g,v_2^g,\ldots v_{\ell}^g$ can be sp-extended to be identical on $Y$.
Let us take those extensions in $T^*$. 

 Let $s_{a+g}$ be defined as $s_{a+g}[A]=z_{a+g}$ where $z_{a+g}$ is a value not used before for $A\in X$, furthermore $s_{a+g}[B]=v_{\ell+1}^g[B]$ for $B\in R\setminus X$. The sp-extensions are given as  $v_q^{g*}[A]=z_{a+g}$ if $v_q^{g*}[A]=\bot$ and $A\in X$, otherwise $v_q^{g*}[A]={v_q^g}'[A]$ for $\ell+1\le q\le k_g$.  Finally, let $s_{a+g}^*[B]={v_1^g}'[B]$ for $B\in R\setminus X$.

For any tuple $t\in T\setminus U$ for which no sp-extension has been defined yet, let us keep its extension in $T'$, that is let $t^*=t'$.
\paragraph{Claim} $T^*\models X\spfd Y$.  Indeed, let $t^1,t^2\in T\cup\{s_1,\ldots ,s_{a+f+1}\}$ be two tuples such that their sp-extensions in $T^*$ agree on $X$, that is $t^{1*}[X]=t^{2*}[X]$. If $t^{1*}[X]$ contains a new value $z_j$ for some $1\le j \le a+f+1$, then by definition of the sp-extensions above, we have $t^{1*}[Y]=t^{2*}[Y]$. Otherwise, either both $t^1,t^2$ are $X$-total,  so again by definition of the sp-extensions above, we have $t^{1*}[Y]=t^{2*}[Y]$, or at least one of them is not $X$-total, and then $t^{1*}={t^1}'$ and $t^{2*}={t^2}'$. But in this latter case using $T'\models X\spfd Y$ we get $t^{1*}[Y]=t^{2*}[Y]$.
\end{proof}
The values $g_3$ and $g_5$ are similarly independent of each other for spFDs as in the case of spKeys.
\begin{theorem}\label{thm:pq-spfd}
For any rational number $0\le\frac{p}{q}<1$ there exists tables with an arbitrarily large number of rows and bounded number of columns that satisfy $g_3(\sigma)-g_5(\sigma)=\frac{p}{q}$ for $\sigma\colon X\spfd Y$.
\end{theorem}
\begin{proof}
Consider the following table $T$.
\begin{table}
   \[T=\begin{array}{cc|c}
      \hline
\multicolumn{2}{c|}{\mathbf{X}} &\mathbf{Y}\\

      X_1              & X_2        &  
\\\hline
         1  &1&1\\
         1  & 2& 2 \\
         \vdots&\vdots&\hfil\vdots\hfil\\
         1&b&\hfil b\hfil\\
      
         \bot  &\bot &b+1\\
         \bot  &\bot&b+2\\
         \vdots&\vdots&\vdots\\
         \bot  &\bot& b+x
     \\ \hline
   \end{array} \]
   \caption{$g_3-g_5=\frac{p}{q}$}\label{eq:spfdg3>g5}
\end{table}
We clearly have $g_3(X\spfd Y)=\frac{x}{x+b}$ for $T$ as all tuples with \nul s must be removed. On the other hand, by adding new tuples and so extending the active domains, we need to be able to make at least $x+b$ pairwise distinct combinations of $X$-values. If $y$ tuples are added, then we can extend the active domains to the sizes $|VD_1|=y+1$ and $|VD_2|=y+b$. Also, if $y$ is the minimum number of tuples to be added, then 
\begin{equation}
    g_3(X\spfd Y)-g_5(X\spfd Y)=\frac{x-y}{x+b}=\frac{p}{q}
\end{equation}
if $cp=x-y$ and $cq=x+b$ for some positive integer $c$. From here $y=x-cp$ and $y+b=c(q-p)$
Thus, what we need is 
\begin{equation}\label{eq:g3g5spfdupper}
 (y+1)(y+b)=(y+1)c(q-p)\ge cq   
\end{equation}
and, to make sure that $y-1$ added tuples are not enough,
\begin{equation}\label{eq:g3g5spfdlower}
    y(y+b-1)=y(c(q-p)-1)\le cq-1.
\end{equation}
Easy calculation shows that (\ref{eq:g3g5spfdupper}) is equivalent with $y\ge\frac{p}{q-p}$, so we take $y=\left\lceil\frac{p}{q-p}\right\rceil$. On the other hand, (\ref{eq:g3g5spfdlower}) is equivalent with $y\le\frac{cq-1}{c(q-p)-1}$. Now, similarly to Case~3 of the proof of Theorem~\ref{thm:spkbounded} observe that $\frac{cq-1}{c(q-p)-1}\rightarrow\infty$ as $c\rightarrow\infty$, so, if $c$ is large enough, then (\ref{eq:g3g5spfdlower}) holds.
\end{proof}
\subsection{Semantic Comparison of $g_3$ and $g_5$}

In this section, we compare the $g_3$ and $g_5$ measures to analyze their applicability and usability for different cases. The goal is to specify when it is semantically better to consider adding or removing rows for approximation for both spFDs and spKeys.

Considering the teaching table in Table \ref{add_or_remove_1}, we have the two strongly possible constraints $Semester \, TeacherID \spfd CourseID$ and $\spk{Semester \, TeacherID}$. It requires adding one row so that $asp_a^+ \left\langle Semester \, TeacherID \right\rangle$ = $_a^+ Semester \, TeacherID \spfd CourseID$. But on the other hand, it requires removing 3 out of the 6 rows. Then, it would be more convenient to add a new row rather than removing half of the table, which makes the remaining rows not useful for analysis for some cases.

Adding new tuples to satisfy some violated strongly possible constraints ensures that we make the minimum changes. In addition to that, in the case of deletion, some active domain values may be removed. There are some cases where it may be more appropriate to remove rather than add tuples, however. This is to preserve semantics of the data and to avoid using values that are out of the appropriate domain of the attributes while adding new tuples with new unseen values. For example, Table \ref{add_or_remove_2} represents the grade records for some students in a course that imply the key ($Name$, $Group$) and the dependency  $Points \, Assignment \fd Result$, while both of $\spk{$Name$ $Group$}$ and $Points \, Assignment \spfd Result$ are violated by the table. Then, adding one tuple with the new values (Dummy, 3, 3, Maybe, Hopeless) is enough to satisfy the two strongly possible constraints, while they can also be satisfied by removing the last two tuples. However, it is not convenient to use these new values for the attributes, since they are probably not contained in the intended domains. Hence, removing two tuples is semantically more acceptable than adding one tuple. 

\setlength{\tabcolsep}{7pt}
\renewcommand{\arraystretch}{1.3}

\begin{table}[ht]
	\centering
   \begin{tabular}{ccc}
\hline
\textbf{Semester} & \textbf{TeacherID} & \textbf{CourseID} \\ \hline
First           & 1              & 1           \\
$\bot$         & 1              & 2        \\
First          & 2              & 3       \\
$\bot$          & 2              & 4       \\
First        & 3              & 5      \\
$\bot$          & 3         & 6   
\\ \hline
\end{tabular}
\caption{Incomplete teaching table}
\label{add_or_remove_1}
\end{table}
  \begin{table}[ht]
    \centering
    \begin{tabular}{ccccc}
\hline
\textbf{Name} & \textbf{Group} & \textbf{Points} & \textbf{Assignment} & \textbf{Result} \\ \hline
Bob           & 1              & 2               & Submitted           & Pass            \\
Sara          & 1              & 1               & Not Submitted       & Fail            \\
Alex          & 1              & 2               & Not Submitted       & Fail            \\
John          & 1              & 1               & Submitted           & Pass            \\
$\bot$        & 1              & 1               & $\bot$              & Retake
  \\
Alex          & $\bot$         & 2               & $\bot$              & Retake
\\ \hline
\end{tabular}
\caption{Incomplete course grading table}
\label{add_or_remove_2}
\end{table}

If $g_3$ is much larger than $g_5$ for a table, it is better to add rows than remove them. Row removal may leave only a short version of the table which may not give a useful data analysis, as is the case in Table~\ref{eq:spfdg3>g5}. If $g_3$ and $g_5$ are close to each other, it is mostly better to add rows, but when the attributes' domains are restricted to a short range, then it may be better to remove rows rather than add new rows with "noise" values that are semantically not related to the meaning of the data, as is the case in Table \ref{add_or_remove_2}.
\section{Strongly possible tuple generating integrity constraints}\label{sec:tuple-generating}
\subsection{Missing Values and MVD}
Investigations of strongly possible worlds can be extended for other types of dependencies, like multivalued dependencies and cross joins.
Note that $T\models X\spmvd Y$ if and only if for each value in $X$, all the possible combinations of $Y$ and $T\setminus XY$ are there in some strongly possible world of $T$.
Then, there should be at least $a*b$ tuples for each value $v \in X$, where $a$ and $b$ are the number of distinct values shown on $Y$ and $T \setminus XY$ respectively on those tuples. 
To check if an FD $X\spfd Y$ is valid in a relation $R$, we only have to check its validity in the projection $R[XY]$; the validity does not depend on the values of the other attributes. On the other hand, the validity of $X\spmvd Y$ in $R$ depends on the values of all attributes and cannot be checked in $R[XY]$ only. 

 In \cite{lien1982equivalence}, Lien defined the NMVD over $R$ as a statement $X\mvd Y$ and considers that the left-hand side of the dependency is $X$-total. According to this, MVD $X\mvd Y$ is satisﬁed by a relation $R$, denoted by $R \models X\mvd Y$, if and only if for every two tuples $t_1, t_2$ in the relation, the following is true: if $t_1$ and $t_2$ are $X$-total and $t_1[X] = t_2[X]$, then there is some tuples $t$ such that $t[XY] = t_1[XY]$ and $t[X(R \setminus Y)] = t_2[X(R \setminus Y)]$.
The assumption of having a limited set of values to be replaced with any \nul\ occurrence makes the satisfaction of spMVD different compared with NMVD by Lien. In Figure \ref{g3-g5-analysis}, Table \ref{ta} satisfies the spMVD $X \spmvd Y$, as the spWolrd obtained by replacing the \nul with either 1 or 2 will satisfy the MVD $X \mvd Y$; but it does not satisfy the NMVD $X \mvd Y$, as it violates the NMVD condition. The same for Table \ref{tb}, the spMVD $X \spmvd Y$ is satisfied by replacing the \nul with 2, and the NMVD $X \mvd Y$ condition is violated. On the other hand, Table \ref{tc} violates the $X \spmvd Y$ and satisfies the NMVD $X \mvd Y$. These cases show that the spMVDs and the NMVDs are independent of each other where a table may satisfy the NMVD and violates the spMVD or vice versa. 
\begin{figure*}[ht]
\centering
\begin{minipage}{1\textwidth}
	\begin{subfigure}[t]{.25\textwidth}
	\centering
    \begin{tabular}{ccc}
    \hline
    \textbf{X} & \textbf{Y} & \textbf{Z} \\ \hline
    1  & 1 & 1       \\
    1  & 1 & 2       \\
    1  & 1 & $\bot$  \\ \hline
    \end{tabular} 
    \caption{}
    \label{ta}
    \end{subfigure}  
\hfill
	\begin{subfigure}[t]{.25\textwidth}
	\centering
    \begin{tabular}{cccc}
    \hline
    \textbf{X} & \textbf{Y} & \textbf{Z} & \textbf{V}\\ \hline
    2  & 1 & 1 & 1       \\
    2  & 2 & 1 &  2   \\
    2  & 2 & 1 &  1   \\
    2  & 1 & 1 &  $\bot$       \\\hline
    \end{tabular} 
    \caption{}
    \label{tb}
	\end{subfigure} 
\hfill
    \begin{subfigure}[t]{.25\textwidth}
    \centering
    \begin{tabular}{ccc}
    \hline
    \textbf{X} & \textbf{Y} & \textbf{Z} \\ \hline
    1  & 1 & 1       \\
    2  & 2 & 2       \\
    $\bot$  & 3 & 3       \\ \hline
    \end{tabular} 
    \caption{}
    \label{tc}
	\end{subfigure}  
\end{minipage}
\caption{$spMVd$ and $NMVD$ Tables}
\label{g3-g5-analysis}
\end{figure*}
\begin{proposition}\label{spfd-spmvd}
If $T \models X\spfd Y$ then $T \models X\spmvd Y$. 
\end{proposition}
This is true as in the strongly possible world $T^{\prime}$ of $T$ that satisfies $X\fd Y$, then $T^{\prime} \models X \mvd Y$ and then $T \models X \spmvd Y$.

It is possible that $X\spmvd Y$ is
not valid in $R$ but it is valid in  $U$, where $U \subset R$, as in Table \ref{tbl}:
\begin{table}[ht]
    \centering
    \begin{tabular}{cccc}
\hline
\textbf{X} & \textbf{Y} & \textbf{Z} & \textbf{V} \\ \hline
1  & 1 & 1  & 1       \\
$\bot$  & 2 & 1  & 2       \\ \hline
\end{tabular}
\caption{$X\spmvd Y$ holds in $U=R\setminus V$}\label{tbl}
\end{table}
Where $X \spmvd Y$ holds in $R\setminus V$ but not in $R$. Furthermore, $X\spmvd Y$ is valid in $R$ if and only if
$X\spmvd Y\setminus X$ is valid in $R$. This property of spMVD follows from the ordinary MVD properties. 

In contrast with the ordinary MVD's, spMVD's don't provide a necessary and sufficient condition for a relation to be decomposable into two of its projections without loss of information. As it shown in Table \ref{dec} with the spMVD $X\spmvd Y$ that is not possible to be decomposed in a lossless way into $XY$ and $XZ$.
\begin{table}[H]
    \centering
\begin{tabular}{ccc}
\hline
\textbf{X} & \textbf{Y} & \textbf{Z}  \\ \hline
1 & 1 & 1         \\
1 & $\bot$ & 2         \\ \hline
\end{tabular}
\caption{Decomposition violation table}
\label{dec}
\end{table}
\subsection{Multivalued Dependencies Approximation}
This section presents an approximation approach that measures how close an spMVD $X\spmvd Y$ is satisfied in a table $T$ if the ordinary MVD $X\mvd Y$ is voilated by $T$. We use an assumption similar to what we presented in Sections ~\ref{spkeySec} and ~\ref{spfdSec}, where the  total part of the table satisfies the MVD $X\spmvd Y$. And we also employee the $g_3$ and $g_5$ measures so that we remove or add tuples, respectivly, to satisfy the violated spMVD. The two measures, $g_3$ and $g_5$ are defined for spMVD in Definitions ~\ref{spmvd-approx-rmv} and ~\ref{spmvd-approx-add}.
\begin{definition}\label{spmvd-approx-rmv}
    For the attribute sets $X$ and $Y$, $\sigma:X \spmvd Y$ is a remove-approximate strongly possible multivalued dependency of ratio $a$ in a table $T$, in notation \\ $T \models \approx_a^- X \spmvd Y$, if there exists a set of tuples $S$ such that the table $T\setminus S \models X \spmvd Y$, and $|S|/|T|\le a$. Then, $g_3(\sigma)$ is the smallest $a$ such that $T \models \approx_a^-\sigma$ holds.
\end{definition}
\begin{definition}\label{spmvd-approx-add}
    For the attribute sets $X$ and $Y$, $\sigma:X \spmvd Y$ is an add-approximate strongly possible multivalued dependency of ratio $b$ in a table $T$, in notation $T \models \approx_b^+X \spmvd Y$, if there exists a set of tuples $S$ such that the table $T\cup S \models X \spmvd Y$, and $|S|/|T|\le b$. Then, $g_5(\sigma)$ is the smallest $b$  such that $T \models \approx_b^+\sigma$ holds. 
\end{definition}
Where the measures $g_3(\sigma)$ and $g_5(\sigma)$ represent the approximation which is the ratio of the number of tuples
needed to be removed or added, respectively, over the total number of tuples so that the table $T$ satisfies $X \spmvd Y$.

Table ~\ref{fig:MVD$g3moreg5$} shows that we need to remove at least 2 tuples (either the first or the last two tuples) st the the table satisfies $x\spmvd Y$. While adding only one tuple (a tuple with a value "2" on $X$ and $\bot$ on $Y$ and $Z$) is enough to satisfy the constraint, then in this case $g_3 > g_5$. On the other hand, in Table \ref{fig:MVD$g3lessg5$}, it is enough to remove only the last tuple to satisfy the spMVD $X\spmvd Y$ but it required to add at least 2 tuples (2,1,2,1) and (2,2,1,1), then $g_3 < g_5$. And it is also possible that $g_3$ = $g_5$ as it is shown in Table ~\ref{fig:MVD$g3eqg5$} where removeing the last tuple or adding the tuple (1,2,1) is enough. Hence, unlike spFd, where it is always true that $g3 \geq g5$, here in the case of spMVD $g3$ and $g5$ are independent on each other.
\\

\begin{minipage}[c]{0.33\textwidth}
	\centering
        \begin{tabular}{cccc|c|c}
        \hline
        \multicolumn{4}{c|}{\textbf{X}} & \textbf{Y} & \textbf{Z} \\ \hline
$\bot$      & 1      & 1     & 1     & 1          & 1          \\
1      & $\bot$      & 1     & 1     & 1          & 1          \\
1      & 1      & $\bot$     & 1     & 2          & 2          \\
1      & 1      & 1     & $\bot$     & 2          & 3          \\
         \hline
\end{tabular}
        		\captionof{table}{MVD \\ Approximation: $g_3 > g_5$}
                \label{fig:MVD$g3moreg5$}
\end{minipage}
\begin{minipage}[c]{0.33\textwidth}
            	\centering
            		\begin{tabular}{c|c|cc}
\hline
\textbf{X} & \textbf{Y} & \multicolumn{2}{c}{\textbf{Z}} \\ \hline
1          & 1          & 1              & 1             \\
1          & 1          & 2              & 1             \\
2          & 1          & 1              & 1             \\
2          & 2          & 2              &  $\bot$       \\ \hline
\end{tabular}
        		\captionof{table}{MVD \\ Approximation: $g_3 < g_5$}
                \label{fig:MVD$g3lessg5$}	
            	\end{minipage}
\begin{minipage}[c]{0.33\textwidth}
            	\centering
            		\begin{tabular}{c|c|c}
              \hline
        			\textbf{X} & \textbf{Y} & \textbf{Z} \\ \hline
1      & 1      & 1      \\
1      & 1      & 2      \\
1      & 2      & $\bot$  \\
         \hline
\end{tabular}
        		\captionof{table}{MVD \\ Approximation: $g3 = g_5$}
                \label{fig:MVD$g3eqg5$}
\end{minipage}

\subsection{Cross joins}
In the present section, some results on cross joins are collected. 
Hannula et.al. \cite{hannula2022discovery} show that the problem of deciding whether there is a cross join that holds on a given table is not only
NP-complete but W[3]-complete in its most natural parameter, namely its arity. In the proof it is used that if $X$ and $Y$ are given, then checking whether $X \times Y$ holds in table $r$ can easily be done in polynomial time. 

However, if $r$ is incomplete, then the same statement does not hold in general. First, we give a simple case that already illustrates in a sense the reason of possible intractability.
\begin{proposition}\label{prop:spcj-singular}
    Let $R$ be a relational schema, $X,Y\in R$ be attributes, $T$ be an incomplete table over $R$. Let the problem \textbf{Singular spCJ} be defined as\newline
    \textbf{Input:} $R$, $X,Y\in R$ and $T$.\newline
    \textbf{Question:} Does $T\models X\times_{sp} Y$ hold?\newline
    Then \textbf{Singular spCJ} is in P.
\end{proposition}
\begin{proof}
 Let the sizes of active domains $VD_X^T$ and $VD_Y^T$ be $x$ and $y$, respectively. A strongly possible world $T'$ satisfies CJ $X \times Y$ if and only if for all $a\in VD_X^T$  and $b\in VD_Y^T$ there exists a tuple $t'\in T'$ such that $t'[X]=a$ and $t'[Y]=b$. Construct a bipartite graph $G=(A,B;E)$ where $A=T$ and $B=VD_X^T\times VD_Y^T$ and $\{t,(a,b)\}\in E$ is an edge iff $t[XY]\sim_w (a,b)$. Clearly the size of $G$ is a polynomial function of the input size and it can be constructed by a single scan of the table $T$. We claim that there exists an spWorld $T'$ such that $T'\models X \times Y$ iff there exists a matching in $G$ covering partite class $B$. Indeed, if $\forall (a,b)\in B \exists t'\in T\colon (t'[X]=a)\land (t'[Y]=b)$, then pick one such tuple $t'$. The edges $\{t,(a,b)\}$ form a matching in $G$ covering $B$, where $t'$ is the extension of $t$ in $T'$. 

 On the other hand, if a matching covering $B$ exists in $G$, then an edge  $\{t,(a,b)\}$ represents a tuple that can be extended to $t'$ with $(t'[X]=a)\land (t'[Y]=b)$. Since $B$ is covered by the matching, every possible value combinations on attributes $X,Y$ occur in the spWorld obtained.

 The existence of a matching covering $B$ can be checked in polynomial time of the input graph $G$, say by the Hungarian Method \cite{lovasz2009matching}.
 \end{proof}
 On the other hand, if $X,Y$ can be arbitrary subsets of the schema $R$, then checking whether $T\models X\times_{sp} Y$ holds is hard. 
 Let the problem \textbf{General spCJ} be defined as\newline
 \textbf{Input:} Schema $R$, sets of attributes $X,Y\subseteq R$ and table $T$ over $R$.\newline
    \textbf{Question:} Does $T\models X\times_{sp} Y$ hold?\newline
Let  the \emph{weak similarity graphs} of  $X$ and $Y$, $G_X^T$ and $G_Y^T$ be defined as $G_X^T=(T,E_X^T)$ (resp.  $G_Y^T=(T,E_Y^T)$ where $\{t_1,t_2\}\in E_X^T)$ iff $t_1[X]\sim_w t_2[X]$ (resp. $\{t_1,t_2\}\in E_Y^T$ iff $t_1[Y]\sim_w t_2[Y]$). If $T'$ is an spWorld of $T$ and for tuples $t_1'[X]=t_2'[X]=\ldots t_k'[X]$ holds, then $t_1,t_2,\ldots t_k$ form a clique (complete subgraph) of the weak similarity graph $G_X^T$ and the same holds for attribute set $Y$. This implies that $T\models X\times_{sp} Y$ holds iff there exists clique partitions $\mathcal{Q}_X$ and $\mathcal{Q}_Y$ of the weak similarity graphs $G_X^T$ and $G_Y^T$, respectively that $\forall Q\in \mathcal{Q}_X\forall Q'\in\mathcal{Q}_Y\colon Q\cap Q'\neq \emptyset$.
 That is, we are given two graphs, and we need a completely crossing pair of clique partitions of them. This problem seems to be hard. However, to prove that \textbf{General spCJ} is NP-complete one does not need the full generality of the clique partition question.
 \begin{theorem}\label{thm:genspCJ}
     The problem \textbf{General spCJ} is NP-complete.
 \end{theorem}
 The following lemma is needed for the proof of Theorem~\ref{thm:genspCJ}.
 \begin{lemma}\label{lem:allgraphs-are-weaksim}
     Let $G=(V,E)$ be a simple graph with $|V|=k$. Then there exists a table over schema $R$ of $k$ attributes that $G$ is the weak similarity graph of attribute set $R$.
 \end{lemma}
 \begin{proof}
    Consider $R=\{A_1,A_2,\ldots ,A_k\}$ and define tuple $t_i$ for $i=1,2,\ldots ,k$ as follows. 
    \begin{equation}
        t_i[A_{\ell}]=\left\{\begin{array}{rl}
            1 & \text{if }\ell=i \\
            \bot & \text{if }\{v_i,v_{\ell}\}\in E\\
            2 & \text{if }\{v_i,v_{\ell}\}\not\in E
        \end{array}
        \right.
    \end{equation}
    In fact, attribute $A_i$ shows the non-similarities with tuple $t_i$. It is easy to see that in this table $t_i\sim_w t_{\ell}\iff\{v_i,v_{\ell}\}\in E$. 
 \end{proof}
 \begin{proof}[of Theorem~\ref{thm:genspCJ}]
   Problem  \textbf{General spCJ} is in NP, since an spWorld $T'$ satisfying $T'\models X \times Y$ is a good witness, as checking a given cross join for a complete table can be done in polynomial time.

   To prove that \textbf{General spCJ} is NP-hard we prove the 
   \begin{equation}\label{eq:x3c-spCJ}
       3DM\prec \mathbf{General spCJ}
   \end{equation}
   Karp-reduction. Let an input of 3DM be be given as $\mathcal{F}\subseteq B\times C\times D$ where $|B|=|C|=|D|=q$, they are pairwise disjoint. The question is whether there exists a matching $\mathcal{F'}\subseteq\mathcal{F}$ such that $|\mathcal{F'}|=q$ and no two members of $\mathcal{F'}$ agree in any coordinate? Let $\mathcal{F}=\{(b_i,c_i,d_i)\colon i=1,2,\ldots ,k\}$. Construct graph $G=(V,E)$ as follows. $V=B\cup C\cup D\cup\bigcup_{i=1}^k\{a_i^j\colon j=1,2,\ldots ,9\}$ and $E=\bigcup_{i=1}^kE_i$, where $E_i$is the set of $18$ edges shown on Figure~\ref{fig:3DM2spCJ}. This gadget is taken from \cite{johnson1979computers}. Actually, what is shown in the famous book of Garey and Johnson is that there exists a matching $\mathcal{F'}$ if and only if the vertex set of graph $G$ can be partitioned into triangles (complete graphs of size three). Note that the size of $G$ is $|V|=3q+9k$ and $|E|=18k$ that is a polynomial function of the input size of the 3DM problem. An incomplete table $T$ over schema $R=\{A_1,A_2,\ldots ,A_{|V|},Y\}$ is constructed as follows. The tuples correspond to the vertices of $G$ and the part over $X=\{A_1,A_2,\ldots ,A_{|V|}\}$ is given by Lemma~\ref{lem:allgraphs-are-weaksim} so that the weak similarity graph $G_X$ of $X$ is $G$. On the other hand,  let the $Y$-values of tuples corresponding to vertices of $G$ be given  by $1$ for $b_i,a_i^3,a_i^4,a_i^7$ and $2$ for $c_i,a_i^1,a_i^6,a_i^8$, finally $3$ for $d_i,a_i^2,a_i^5,a_i^7$. If some vertices of $B\cup C\cup D$ are not covered by $\mathcal{F}$, that is they are isolated vertices in $G$, then their corresponding tuples have $Y$ value $1$. Our claim is that there exists a matching $\mathcal{F'}\subseteq\mathcal{F}$ if and only if $T\models X\times_{sp} Y$. 

   Assume first, that there exists a matching $\mathcal{F'}\subseteq\mathcal{F}$ say $(b_1,c_1,d_1),\ldots (b_q,c_q,d_q)$. Then on $X$ take weak similarity cliques $b_ia_i^1a_i^2$, $c_ia_i^4a_i^5$, $d_ia_i^7a_i^8$, $a_i^3a_i^6a_i^9$ for $i=1,\ldots ,q$. Furthermore if $(b_j,c_j,d_j)\in\mathcal{F}$ is not in the matching, then take on $X$ the weak similarity cliques $a_j^1a_j^2a_j^3$, $a_j^4a_j^5a_j^6$, $a_j^7a_j^8a_j^9$. These cliques partition the tuples of table $T$. Extend the three tuples in each clique to the same complete tuple on $X$, respectively. Observe that each clique meets all three possible $Y$ values by construction, so the strongly possible world $T'$ obtained satisfies the cross join $X\times Y$.

   On the other hand, assume now that $T\models X\times_{sp} Y$. Observe that cliques of $G$ have size at most three. Let $T'$ be the spWorld that satisfies the cross join $X\times Y$. That means that each $X$ value occurs in tuple with each three possible $Y$-values. These tuples must form a weak similarity clique on $X$, so a triangle in $G$. That is the tuples of $T$ are partitioned into weak similarity cliques of size three, that corresponds to partitioning of the vertex set of $G$ into triangles. This latter one implies that there is a matching  $\mathcal{F'}\subseteq\mathcal{F}$ by \cite{johnson1979computers}.
\begin{figure}
    \centering
    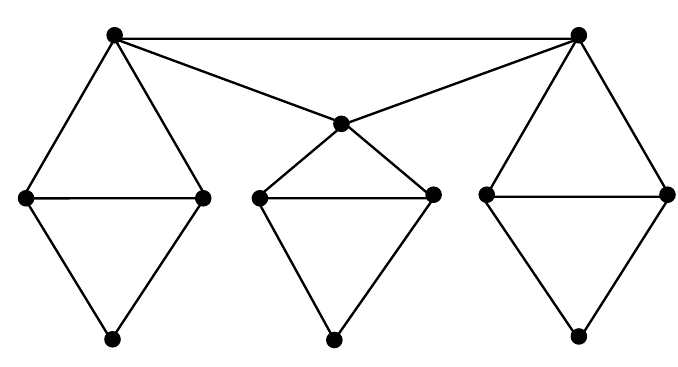
    \caption{Gadget for $3DM\prec \mathbf{General spCJ}$}
    \label{fig:3DM2spCJ}
\end{figure}   
   \end{proof}

\subsection{Cross Join Approximation}
To measure the approximation of $spCJ$, we employ the  measures $g_3$ and $g_5$ to find the minimum number of tuples to remove or add so that all the values combinations of the attribute sets $X$ and $Y$ are there in some strongly possible world. To precisely define $g_3$ and $g_5$ for an spCJ $\sigma\colon X\times_{sp} Y$ one needs only to consider that $X\times_{sp} Y$ is the same as the spMVD $\emptyset\spmvd X$ in the projection of the given table to attribute set $XY$. Note that $g_5$ may take values larger than $1$ in this case. Removing tuples ensures taking off parts of the table where it is not possible to achieve all the strongly possible combinations. Table \ref{fig:$g3moreg5$} shows that removing the last two tuples satisfies the spCJ $TeacherID \times_{sp} CourseID$, as it is not possible in any way to have all the values (1, 2, and 3) on $CourseID$ to appear with the values of $TeacherID$. On the other hand, adding one tuple with a null value for each attribute  satisfies the spCJ. The two measure are independent of each other as it is shown by Table \ref{fig:$g3moreg5$} that $g_3 > g_5$, Table \ref{fig:$g3lessg5$} shows that $g_3 < g_5$, and they can be equal as shown in Table \ref{fig:MVD$g3eqg5$} where $Y\times_{sp}Z$ hols.
\\
\begin{minipage}[c]{0.49\textwidth}
	\centering
        \begin{tabular}{cc}
        			\hline
         TeacherID& CourseID\\ \hline
         1& 1\\
         1& 2\\
         1& 3\\
         2& $\bot$ \\
         2& $\bot$ \\
         \hline
\end{tabular}
        		\captionof{table}{Cross Join Approximation: \\$g_3 > g_5$}
                \label{fig:$g3moreg5$}
\end{minipage}
\begin{minipage}[c]{0.49\textwidth}
            	\centering
            		\begin{tabular}{cc}
        			\hline
         TeacherID& CourseID\\ \hline
         1& 1\\
         1& 2\\
         1& 3\\
         2& $\bot$ \\
         \hline
\end{tabular}
        		\captionof{table}{Cross Join Approximation: \\$g_3 < g_5$}
                \label{fig:$g3lessg5$}	
            	\end{minipage}

The next two theorems show that not only checking the satisfaction of a given spCJ is hard, but finding the approximation measures $g_3$ and $g_5$ are NP-complete.
\begin{theorem}\label{thm:g3spCJ}
    Let problem $\mathbf{spCJ-g_3}$ be defined as follows. \\
    \textbf{Input}: Incomplete table $T$ over schema $R$, attribute sets $X,Y\subseteq R$, rational number $q\in[0,1)$.\\
    \textbf{Question}: Does $g_3(X\times_{sp}Y)\le q$ hold?\\
    Then $\mathbf{spCJ-g_3}$ is NP-complete.
\end{theorem}
\begin{proof}
    To show that $\mathbf{spCJ-g_3}\in$NP the witness is a set $U$ of $k$ tuples to be removed from $T$ and an spWord $T'$ of $T\setminus U$ such that $T'\models X\times Y$ and $\frac{k}{|T|}\le q$.  It is easy to see that this witness can be verified in polynomial time.\\
    To prove that  $\mathbf{spCJ-g_3}$ is NP-hard a Karp-reduction 
    \[
    \mathbf{MaxClique}\prec \mathbf{spCJ-g_3}
    \]
    is presented. Let the input of $\mathbf{MaxClique}$ be graph $G=(V,E)$ and $k\in\mathbb{N}$, with $V=\{v_1,v_2,\ldots ,v_n\}$. Construct a table $T=\{t_1,t_2,\ldots ,t_n\}$ over schema $R=\{X,A_1,A_2,\ldots A_n\}$ as follows. Let 
  $Y=\{A_1,A_2,\ldots A_n\}$,  $t_i[X]=i$ and the restriction $T[Y]$ is constructed by using Lemma~\ref{lem:allgraphs-are-weaksim} so that  the weak similarity graph $G_Y$ is isomorphic to $G$. Furthermore, let $q=1-\frac{k}{n}$. Our claim is that $(G,k)\in \mathbf{MaxClique}\iff g_3(X\times_{sp}Y)\le q$. Indeed, assume  first that $G\cong G_Y$ has a clique $Q\subseteq V$ of size $k$. Then let $T_1=\{t_i\colon v_I\in Q\}$. Since these tuples form a weak similarity clique in $G_Y$, they can be extended to the same value on $Y$. (Note that we might need to use the special symbol $ssymb$ introduced after Definition~\ref{def:spWorld}, as $T_1$ is obtained from $T$ by removing tuples.  Take that spWorld $T_1'$ of $T_1$, it clearly satisfies $T_1'\models X\times Y$. Thus $g_3(X\times_{sp}Y)\le \frac{|V|-|Q|}{|V|}=1-\frac{k}{n}$. 
\end{proof}
\begin{theorem}\label{thm:g5spCJ}
    Let problem $\mathbf{spCJ-g_5}$ be defined as follows. \\
    \textbf{Input}: Incomplete table $T$ over schema $R$, attribute sets $X,Y\subseteq R$, non-negative rational number $q$.\\
    \textbf{Question}: Does $g_5(X\times_{sp}Y)\le q$ hold?\\
    Then $\mathbf{spCJ-g_5}$ is NP-complete.
\end{theorem}
\begin{proof}
    The witness that $\mathbf{spCJ-g_5}\in$NP consists an (incomplete) table $T\subseteq U$ and an spWorld $U'$ of $U$ such that $U'\models X\times Y$ and $\frac{|U|-|T|}{|T|}\le q$. It is easy to see that the witness can be checked in polynomial time. To prove that it is NP-hard, a Karp-reduction
    \[
    \mathbf{3-Color}\prec \mathbf{spCJ-g_5}
    \]
    is presented. Let simple graph $G=(V,E)$ be the input of $\mathbf{3-Color}$, with $V=\{v_1,v_2,\ldots ,v_n\}$. Create a table $T=\{t_1,t_2,\ldots ,t_n\}$ over schema $R=\{X,A_1,A_2,\ldots A_n\}$ as follows. Let 
  $Y=\{A_1,A_2,\ldots A_n\}$,  $t_i[X]=i$ and the restriction $T[Y]$ is constructed by using Lemma~\ref{lem:allgraphs-are-weaksim} so that  the complement $\overline{G_Y}$ of the weak similarity graph $G_Y$ is isomorphic to $G$. Let $q=2$. We claim that $G$ is properly 3-colorable iff $g_5(X\times_{sp}Y)\le 2$.\\
  Assume first that $\chi(G)\le 3$. This means that $V$ can be partitioned into three independent sets, that is $G_Y$ can be partitioned into three cliques. So there exist three possible $Y$-extensions $\underline{y_i}\colon i=1,2,3$ so that each tuple of $T$ can be extended on $Y$ to one of them. Add $2n$ tuples consisting only $\bot$'s to $T$ to obtain table $U$. Let $U'$ be the spWorld of $U$ obtained by extending each $t_i$ to one of $\underline{y_i}\colon i=1,2,3$, and for each $i$ extend two of the added $2n$ tuples to $(i,\underline{y_j}$ where $j$ takes those two values that are not used in the extension of $t_i[Y]$. It is clear that $U'\models X\times Y$.\\
  On the other hand, assume that $g_5(X\times_{sp}Y)\le 2$. This means that there exists a table $T\subseteq U$ and an spWorld $U'$ of $U$ such that $U'\models X\times Y$ and $|U|\le 3n$. Since $U'[X]$ has at least $n$ different values and in total $U'$ has at most $3n$ tuples, $U'[Y]$ can have at most 3 different values if cross join $X\times Y$ holds in $U'$. In particular, tuples $t_i$ have at most three different $Y$-extensions, that is the weak similarity graph $G_Y$ can be covered with at most three cliques, that is $G$ can be covered by at most threee independent sets, which gives $\chi(G)\le 3$.
\end{proof}

\section{Conclusion and Future Directions}\label{cnclosnSec}
 We investigated two approximation measures $g_3$ (the ratio of the minimum number of tuples to be removed) and the newly introduced measure $g_5$ (the ratio of the  minimum number of tuples to be added) for strongly possible constraints in the present paper. Adding tuples is only useful for strongly possible constraints, as the new added tuples may add more values to the active domains of attributes and this may satisfy some strongly possible constraints. While adding new values cannot satisfy violated ordinary constraints or possible/certain dependencies. 

Instead of removing or adding only tuples one may consider removals and additions concurrently, or one might even consider tuple modifications. In case of spKeys if first the additions are performed, then after that by Proposition~\ref{prop:removenontotal} it is
always true that we can remove only non-total tuples; then, instead of any tuple removal, we
may add a new tuple with distinct values. Therefore, mixed modification in that way would
not change the approximation measure, as it is always equivalent to tuples addition only.  However, it is interesting to further investigate the order of removals and additions and its effect on considering the removals to be substituted by additions.
However, for other strongly possible constraints equivalent statement cannot be proven, so  the further research direction of mixed modifications are interesting research problem for spFDs, spMVDs and spCJs.

Modifying a tuple can also be done by removing the tuple and adding a new one with the required modification. So, the research question is still open whether tuple modification can provide better approximation result than adding then removing a tuple. 

We introduced strongly possible versions of multivalued dependencies and cross joins; definitions and properties are provided. We proved that spFDs determines spMVDs as it is the case for ordinary FDs and MVDs. We compared our proposed spMVDs with NMVDs introduced by Lien \cite{lien1979multivalued}, and its shown that NMVDs and spMVDs are independent of each other. An investigation on the satisfaction questions of spCJs is provided. We showed that checking the satisfaction an spCJ of two singular attributes is in P. And also we proved that the problem of checking spCJs in general is NP-complete.  

We extended the approximation measure to other strongly possible integrity constraints.
We proved that the $g_3$ value is less than or equal to the $g_5$ value for a given strongly possible key or FD in a table. Nevertheless, the two measures are independent of each other in the sense that $v = g_3 - g_5$ for any rational value $v$ such that $0\leq v < 1$. On the other hand, we showed that the two measures values are independent of each other for a given spMVD or spCJ in a table.

Some complexity issues were also treated in the present paper continuing the work started in \cite{alattar2020functional,alatar2024}. We showed that the decision problems is $T\models X\times_{sp}Y$ true, does $g_3(X\times_{sp}Y)\le q$ hold and does $g_5(X\times_{sp}Y)\le q$ hold are all NP-complete. 

A future research direction  is given by the fact that we have no fixed order between the $g_3$ and $g_5$ measures for spCJs and spMVDs. What is the complexity to decide whether $g_3(\sigma)<g_5(\sigma)$ for a given $\sigma$ and table $T$ if $\sigma$ is an spCJ or spMVD?
Another interesting question is whether one may get better approximations of spFDs, spCJs and spMVDs if deletions and additions of tuples are both allowed at the same time? This could be a question like edit distance of graphs, that is we calculate the total number of modifications of the table, additions and deletions together. 

\bibliographystyle{plain}
\bibliography{spFD-key_approx}

\end{document}